\documentclass[11pt]{article}

\usepackage[a4paper,margin=1in]{geometry}
\usepackage{amsmath,amssymb,amsthm,mathtools}
\usepackage{enumitem}
\usepackage{microtype}
\usepackage{xcolor}
\usepackage[colorlinks=true,linkcolor=blue!55!black,citecolor=blue!55!black,urlcolor=blue!55!black]{hyperref}
\usepackage[nameinlink,capitalize,noabbrev]{cleveref}
\allowdisplaybreaks[2]
\hypersetup{pdftitle={Spectral width and polynomial degree in perfect state transfer},
  pdfauthor={Xingkun Song}}

\newtheorem{theorem}{Theorem}[section]
\newtheorem{lemma}[theorem]{Lemma}
\newtheorem{proposition}[theorem]{Proposition}
\newtheorem{corollary}[theorem]{Corollary}
\theoremstyle{definition}
\newtheorem{example}[theorem]{Example}
\theoremstyle{remark}
\newtheorem{remark}[theorem]{Remark}

\newcommand{\C}{\mathbb C}
\newcommand{\R}{\mathbb R}
\newcommand{\Z}{\mathbb Z}
\newcommand{\T}{\mathsf T}
\newcommand{\supp}{\operatorname{supp}}
\newcommand{\spec}{\operatorname{Spec}}
\newcommand{\ee}{\mathrm e}
\newcommand{\ii}{\mathrm i}

\title{Spectral width and polynomial degree in perfect state transfer}
\author{
  Xingkun Song$^{1,2}$\thanks{Email:
  \href{mailto:xksong@126.com}{xksong@126.com}}\\[2ex]
  {\small $^{1}$ School of Mathematics and Statistics, Qinghai Minzu University,}\\
  {\small Xining, Qinghai 810007, P.R. China}\\[1ex]
  {\small $^{2}$ Qinghai Institute of Applied Mathematics,}\\
  {\small Xining, Qinghai 810007, P.R. China}
}

\date{}

\begin{document}

\maketitle

\begin{abstract}
We study the minimum time for perfect state transfer under polynomial
Hamiltonians with bounded degree and spectral width.  For a strongly
cospectral pair and width bound $W$, the optimum, when finite, is an
integer multiple of $\pi/W$, determined by integer interpolation with
prescribed parities.  For equally spaced supported eigenvalues with alternating
signs, we give degree bounds under which every minimizer is affine,
and sharp asymptotics for each fixed exact degree.  Near-minimizing
phase polynomials satisfy a quantitative Chebyshev stability estimate.  We determine the optimal
transfer time for every degree bound on hypercubes of odd prime
dimension.  For complementary vertices of $J(2m,m)$, the optimal time
at fixed spectral width grows exponentially in $m$ throughout an
interval of feasible degrees.  We also construct polynomial
Hamiltonians showing that every feasible degree $m-t$ with $t=o(m)$
admits subexponential transfer time.
\end{abstract}



\section{Introduction}

Let $H$ be a real symmetric matrix indexed by a finite set $V$.
Perfect state transfer (PST) from $a$ to $b$ at time $\tau>0$ means
\[
 \exp(-\ii\tau H)e_a=\gamma e_b,\qquad |\gamma|=1.
\]
Throughout the paper, transfer pairs consist of distinct vertices.
Write $H=\sum_{\lambda\in\Phi}\lambda E_\lambda$, where
$\Phi=\spec(H)$, and define
\[
 \Phi_a=\supp_H(e_a)=\{\lambda:E_\lambda e_a\ne0\},\qquad
 \R[H]=\{p(H):p\in\R[x]\}.
\]
The vertices $a,b$ are \emph{strongly cospectral} if their
supports coincide and
$E_\lambda e_b=\sigma_\lambda E_\lambda e_a$, with
$\sigma_\lambda\in\{1,-1\}$, on the common support.

The adjacency model of state transfer arose in quantum communication through
spin networks~\cite{Bose2003,ChristandlEtAl2004,ChristandlEtAl2005}.
Its spectral and arithmetic foundations were developed by
Godsil~\cite{Godsil2011Periodic,Godsil2012,Godsil2012When};
see also~\cite{Kay2010}.  Polynomial Hamiltonians already occur in
long-range spin networks~\cite{JafarizadehSufiani2008},
next-to-nearest-neighbor models~\cite{ChristandlVinetZhedanov2017}, and
weighted graphs in the Johnson scheme~\cite{VinetZhan2020}.
Strong cospectrality and its characterization by polynomial symmetries
are studied in~\cite{GodsilSmith2024,Monterde2022}.

For any real-valued function $f$ on $\Phi$, functional calculus gives
\[
 f(H)=\sum_{\lambda\in\Phi}f(\lambda)E_\lambda.
\]
Since $\Phi$ is finite, interpolation gives $f(H)=p(H)$ for a real
polynomial $p$.  Song~\cite[Proposition~2.2]{SongPowersParity} gives
necessary and sufficient phase conditions for PST under $f(H)$, recalled
in \cref{lem:spectral-pst}; that paper also classifies the
exponents for which a matrix power admits PST.  We study the minimum
transfer time under bounds on spectral width and polynomial degree.

\subsection{Spectral width and integer interpolation}

For a real symmetric matrix $K$, write
\[
 W_a(K)=\max\supp_K(e_a)-\min\supp_K(e_a),\qquad
 W(K)=\max\spec(K)-\min\spec(K).
\]
We call $W_a(K)$ the local spectral width at $a$ and $W(K)$ the
global spectral width, abbreviated to local width and global width.
If $K$ has PST from $a$ to $b$ at time $\tau$, then
$\tau W_a(K)\geq\pi$;
see~\cite{NessAlbertiSagi2022}.

Fix distinct strongly cospectral vertices $a,b$ for $H$ and a width
bound $W>0$.  Write
\[
 m_a(x)=\prod_{\lambda\in\Phi_a}(x-\lambda),\qquad
 r_-(\lambda)=\frac{1-\sigma_\lambda}{2}\quad(\lambda\in\Phi_a),
\]
where $\deg r_-<|\Phi_a|$, and put $\delta=\deg r_-$.
The \emph{local degree} of $p$ is the degree of its unique representative
of degree less than $|\Phi_a|$ modulo $m_a$.  Let $T_d(W)$ denote the minimum transfer time among
polynomial Hamiltonians of local degree at most $d$ and local width
at most $W$, with $T_d(W)=\infty$ if no such Hamiltonian admits PST.

\begin{theorem}\label{thm:optimal-transfer-time}
The minimum transfer time in $\R[H]$ subject to
$W_a(p(H))\leq W$, or to $W(p(H))\leq W$, is
$\pi/W$.  The polynomials attaining this time under the local
width constraint are precisely
\begin{equation}\label{eq:all-optimal-polynomials}
 p\equiv c+\varepsilon W r_-\pmod{m_a},
 \qquad c\in\R,\quad\varepsilon\in\{1,-1\}.
\end{equation}
Such a polynomial also satisfies the global width constraint if and
only if its values at unsupported eigenvalues lie between $c$ and
$c+\varepsilon W$.
Each of the two supported eigenvalues of a minimizing $p(H)$ has
spectral weight $1/2$ at $a$.
\end{theorem}

Enumerate
$\Phi_a=\{\lambda_0,\ldots,\lambda_{s-1}\}$ and write
$\sigma_j=\sigma_{\lambda_j}$.  For $0\leq d<s$, set
\begin{align}
 \mathcal V_d
 &=\{(q(\lambda_j)-q(\lambda_0))_{j=0}^{s-1}:
                                q\in\R[x],\ \deg q\leq d\},
 \label{eq:eval-space}\\
 \mathcal Z_d
 &=\{z\in\mathcal V_d\cap\Z^s:
                 z_j\equiv(1-\sigma_j\sigma_0)/2\pmod2\}.
 \label{eq:parity-lattice}
\end{align}
The zero polynomial is allowed in the definition of $\mathcal V_d$.
All vectors in $\mathcal V_d$ have $z_0=0$.  Elements of $\mathcal Z_d$
are integer phase data, normalized at $\lambda_0$.  The degree of
phase data means the degree of their unique interpolating polynomial
of degree less than $s$ on the supported eigenvalues.  A degree bound is feasible when
$\mathcal Z_d\ne\varnothing$.  Define
\[
 R_d=\min_{z\in\mathcal Z_d}(\max_jz_j-\min_jz_j),
\]
with $R_d=+\infty$ when $\mathcal Z_d$ is empty.

\begin{theorem}\label{thm:degree-width}
For $0\leq d<|\Phi_a|$,
\begin{equation}\label{eq:degree-optimum}
 T_d(W)=\frac{\pi R_d}{W}.
\end{equation}
If $\mathcal Z_d\ne\varnothing$, its minimum range is attained and
$R_d$ is a positive integer.  Moreover,
\[
 R_d=1\quad\Longleftrightarrow\quad d\geq\delta.
\]
Thus PST with local degree at most $d$ is possible exactly when
$\mathcal Z_d\ne\varnothing$; if $d<\delta$, its time is at least
$2\pi/W$.
\end{theorem}

For $H=A(G)$, one has $\delta\geq\operatorname{dist}_G(a,b)$ by
\cref{cor:distance-degree}, with $\delta=n-1$ for the endpoints of
$P_n$.  Moreover, if $\deg p\leq d$, then $(p(A))_{xy}=0$ whenever
$\operatorname{dist}_G(x,y)>d$.

For the hypercube $Q_p$, with $p$ an odd prime, every degree bound $1\leq d<p$ gives minimum transfer time
$p\pi/W$, while the degree bound $p$ gives $\pi/W$; see
\cref{thm:prime-cube}.  At even distance $D$ in a bipartite graph,
strong cospectrality and support size $D+1$ instead imply
$T_{D-1}(W)=2\pi/W$; see \cref{thm:graph-costs}.

\subsection{Equally spaced spectra}

In this subsection, suppose
\[
 \Phi_a=\{\theta+hj:0\leq j\leq n\},\qquad h\ne0,
 \qquad \sigma_j=\sigma_0(-1)^j.
\]
Throughout, $\log$ denotes the natural logarithm unless a base is specified.
Write $L_d=\operatorname{lcm}(1,\ldots,d)$ for $d\geq1$.
For $0<u<1$, put
\begin{equation*}
 \psi(u)=(1+u)\log(1+u)-(1-u)\log(1-u)-2u\log(4u),
\end{equation*}
and let $\rho_0\in(1/2,1)$ be its unique zero in that interval;
numerically, $\rho_0\approx0.6298694$.
Write $a_0=-\psi'(\rho_0)>0$.

For $1\leq d\leq n$, also set
\[
 \mathcal D_{n,d}=\frac{4}{n+1}
       \frac{\binom{n+d+1}{2d+1}}{\binom{2d}{d}}.
\]

\begin{theorem}\label{thm:arithmetic-width}
For $1\leq d\leq n$,
\begin{equation}\label{eq:arithmetic-width}
 R_d\geq
 \max_{\substack{1\leq m\leq n\\m\text{ odd}}}
 \frac{m}{\gcd(m,L_d)}.
\end{equation}
Furthermore, $R_d=n$ and all minimizing local representatives are
\begin{equation}\label{eq:affine-minimizers}
 q(x)=c\pm\frac{W}{n|h|}x,\qquad c\in\R,
\end{equation}
provided either $n\geq3d$, or
\begin{equation}\label{eq:cube-finite-rigidity}
 n\geq6,\qquad d>n/3,\qquad4\mathcal D_{n,d}>n^2.
\end{equation}
For every $\eta>0$, the same conclusion holds for all sufficiently
large $n$ whenever
\begin{equation}\label{eq:cube-log-rigidity}
 1\leq d\leq\rho_0n-\left(\frac{3}{a_0}+\eta\right)\log n.
\end{equation}
\end{theorem}

In particular, for each fixed $0<\beta<\rho_0$, every degree bound
$1\leq d\leq\beta n$ admits only affine minimizers for all sufficiently
large $n$.  When $\Phi_a=\spec(H)$, local and global width agree,
so all statements in this subsection also hold under global width.

For hypercubes, $h=-2$ and the full spectrum is supported.  The first
strict improvement over the affine transfer time therefore requires
a degree of at least $(\rho_0-o(1))n$.  Throughout the degree range
\eqref{eq:cube-log-rigidity}, every minimizing Hamiltonian has only
nearest-neighbor couplings and a scalar diagonal term.  Whether
$\rho_0$ is the asymptotic threshold remains open.

For $2\leq k\leq n$, let
$S_{n,k}(W)$ be the minimum transfer time with local width at most $W$
and local degree exactly $k$, and set
\[
 C_k^*=\frac{1}{2^{2k-3}k!}.
\]
Let $\T_k(\cos t)=\cos(kt)$.  For a degree-$k$ polynomial $f_n$ with
$f_n(j)\in\Z$ and $f_n(j)\equiv j\pmod2$ on $0\leq j\leq n$, write
\[
 R_n=\max_j f_n(j)-\min_j f_n(j),\qquad
 b_n=\frac{\max_j f_n(j)+\min_j f_n(j)}2,\qquad
 F_n(t)=\frac{2(f_n(nt)-b_n)}{R_n}.
\]

\begin{theorem}\label{cor:exact-degree-growth}
For each fixed $k\geq2$,
\begin{equation}\label{eq:exact-degree-growth}
 S_{n,k}(W)=\frac{\pi}{W}
 \bigl(C_k^*n^k+O_k(n^{k-1})\bigr),
\end{equation}
and, for every $n\geq2$,
\begin{equation}\label{eq:quadratic-exact-degree}
 S_{n,2}(W)=\frac{\pi}{W}\left\lceil\frac n2\right\rceil^2.
\end{equation}
There is $C_k>0$ such that, for all sufficiently large $n$ and
$0\leq\varepsilon\leq1/4$,
\begin{equation}\label{eq:chebyshev-stability}
 R_n\leq(1+\varepsilon)C_k^*n^k
 \quad\Longrightarrow\quad
 \min_{\epsilon=\pm1}
 \|\epsilon F_n-\T_k(2\,\cdot-1)\|_{\infty,[0,1]}
 \leq C_k(\varepsilon+n^{-1}).
\end{equation}
Consequently, if $R_n\sim C_k^*n^k$, there are signs $\epsilon_n$ such that
\begin{equation}\label{eq:chebyshev-profile}
 \epsilon_nF_n(t)\longrightarrow\T_k(2t-1)
 \quad\text{uniformly on }[0,1].
\end{equation}
\end{theorem}

The limits in \cref{cor:exact-degree-growth} are taken with $k$ fixed.

\subsection{Johnson graphs}

The least degree admitting PST between complementary vertices of
$J(2m,m)$ is determined by
Vinet and Zhan~\cite[Theorems~3.6 and~3.7]{VinetZhan2020}.

For complementary vertices of $J(2m,m)$, put
$q=2^{\lfloor\log_2m\rfloor}$ and $\Gamma=(9+\sqrt{17})/8$.  Define
\begin{equation}\label{eq:johnson-range-bound}
 \Lambda_{m,d}=
 \left(\frac{2}{m+1}
 \frac{\binom{2m+2d+2}{4d+1}}{\binom{4d}{2d}}\right)^{1/2},
\end{equation}
\[
 U_{m,q}^2=\sum_{k=1}^q
 \frac{\binom{2m+2k+2}{4k+1}}{\binom{4k}{2k}},\qquad
 B_{m,q}=\sum_{r=0}^{\log_2q}\binom{m+2^r}{2^{r+1}}.
\]
For $d\geq q$, let $\varepsilon_q=1$ and $\varepsilon_d=2$ for $d>q$.
The full spectrum is supported, so local and global width agree
throughout this subsection.

\begin{theorem}\label{thm:johnson-exponential}
For complementary vertices of $J(2m,m)$, one has $R_d=\infty$ for
$d<q$ and $R_m=1$.  For $q\leq d\leq m$,
\begin{equation}\label{eq:johnson-general-bounds}
 \varepsilon_d\Lambda_{m,d}\leq R_d\leq\min\{U_{m,q},B_{m,q}\},
\end{equation}
where $\log U_{m,q}=m\log\Gamma+O(\log m)$.
In particular, if $m=2q-1$ and $q=2^a$ with $a\geq1$, then
\begin{equation}\label{eq:johnson-exponential}
 \left(\frac{2\binom{6q}{2q}}
 {(4q+1)\binom{4q}{2q}}\right)^{1/2}
 \leq R_q\leq U_{2q-1,q}.
\end{equation}
\end{theorem}

The lower bound in \eqref{eq:johnson-exponential} is
$\Theta(q^{-1/2}(27/16)^q)$, while
$\log U_{2q-1,q}=2q\log\Gamma+O(\log q)$.
Hence, along this family,
\[
 T_q(W)=\frac{\pi}{W}\exp(\Theta(m)),\qquad
 T_m(W)=\frac{\pi}{W}.
\]

\begin{theorem}\label{thm:johnson-degree-loss}
For $1\leq t\leq m-q$,
\begin{equation*}
 R_{m-t}\leq6L_t\binom{2m+1}{t}\binom{4m+2}{t},
 \qquad \log R_{m-t}=O\!\left(t\log\frac{em}{t}\right).
\end{equation*}
If $m$ is not a power of two, then $R_{m-1}\leq3m-2$.
\end{theorem}

Thus $R_{m-t}=\exp(o(m))$ for every feasible degree loss $t=o(m)$,
and $R_{m-t}=O_t(m^{2t})$ for every fixed feasible degree loss.
Here a degree loss $t$ is feasible when $m-t\geq q$.
The constructions are given in
\cref{prop:johnson-deficit,prop:johnson-linear}.

\begin{corollary}\label{cor:johnson-window}
For every $\beta\in(1/2,\rho_0)$ there is $c_\beta>0$ such that,
for all sufficiently large $m$ and every integer $q\leq d\leq\beta m$,
\begin{equation}\label{eq:johnson-window}
 R_d\geq c_\beta m^{-1/2}\exp(m\psi(\beta)).
\end{equation}
Consequently,
$T_d(W)=(\pi/W)\exp(\Theta_\beta(m))$ uniformly over this degree
interval whenever it is nonempty.
In particular, for any sequence of polynomial Hamiltonians on
$J(2m,m)$ transferring between complementary vertices, with local
degrees $d_m$, widths $W_m>0$, and transfer times $\tau_m$,
\[
 \log(\tau_mW_m/\pi)=o(m)
 \quad\Longrightarrow\quad
 \liminf_{m\to\infty}\frac{d_m}{m}\geq\rho_0.
\]
The same statements hold along any sequence of diameters tending to infinity.
\end{corollary}

For example, $\beta=3/5$ gives $\psi(3/5)\approx0.0679596>0$.
When $m=2q-1$, the interval from $q$ to $\lfloor3m/5\rfloor$
contains a number of degrees proportional to $m$.

The exponential bounds are in the diameter $m$; the graph has
$\binom{2m}{m}$ vertices.  For comparison, on $Q_m$ an affine
Hamiltonian gives $T_q(W)\leq m\pi/W$.  Thus the distance $m$ and
support size $m+1$ do not determine the order of $T_q(W)$.

For strongly cospectral vertices at distance $D$ with $D+1$ supported
eigenvalues, \cref{thm:weight-partition} characterizes
$T_{D-1}(W)=2\pi/W$ by a subset of one sign class with spectral
weight $1/4$.  For $1\leq d<D-1$, \cref{thm:weight-partition}(i)
replaces the weight equality by equalities of spectral moments.
Part \textup{(ii)} of that theorem characterizes $T_d(W)\leq3\pi/W$
by the signed moment identity \eqref{eq:signed-quarter}.
Together with explicit interpolating polynomials, these characterizations
give the exact times on $J(6,3)$ and $J(10,5)$ in
\cref{prop:johnson-optima}.

\Cref{sec:spectral} proves the spectral interpolation and optimization
statements.  \Cref{sec:degree-graphs} characterizes degree constraints
by spectral moment equations and applies the resulting weight partitions
to bipartite graphs.  \Cref{sec:equally-spaced} proves the
results for equally spaced spectra.  \Cref{sec:johnson} gives the Johnson bounds and
constructions.  \Cref{sec:conclusion} discusses open questions.

\section{Spectral interpolation and time optimization}\label{sec:spectral}
\label{sec:time-optimality}

The following lemma is \cite[Proposition~2.2]{SongPowersParity};
see also \cite[Theorem~3.1]{SongPowersParity}.
For any vector $x$, put
$\supp_H(x)=\{\lambda:E_\lambda x\ne0\}$.

\begin{lemma}\label{lem:spectral-pst}
Let $H$ be real symmetric and $f:\spec(H)\to\R$.
Then $f(H)$ has PST from $a$ to $b$ at time $\tau>0$ with phase $\gamma$
if and only if $a,b$ are strongly cospectral for $H$ and
\begin{equation}\label{eq:spectral-phases}
 \ee^{-\ii\tau f(\lambda)}=\gamma\sigma_\lambda
 \qquad(\lambda\in\Phi_a).
\end{equation}
\end{lemma}

\begin{proof}
Applying $E_\lambda$ to the transfer equation gives
$\ee^{-\ii\tau f(\lambda)}E_\lambda e_a=\gamma E_\lambda e_b$.
The projections vanish together.  When nonzero, they are real vectors
related by a scalar of modulus one; that scalar is therefore $1$ or $-1$.
Conversely, summing the projected equations gives PST.
\end{proof}

\begin{proposition}\label{prop:local-polynomial-residue}
Let $H$ be real symmetric.  For $p,q\in\R[x]$, the following are equivalent:
\[
 p(H)e_a=q(H)e_a;\qquad
 p|_{\Phi_a}=q|_{\Phi_a};\qquad m_a\mid(p-q).
\]
When they hold, $p(H)$ and $q(H)$ agree on the complex cyclic subspace
$\mathcal K_a=\operatorname{span}_{\C}\{H^je_a:j\geq0\}$ and
\[
 \exp(-\ii t p(H))e_a=\exp(-\ii t q(H))e_a\qquad(t\in\R).
\]
Every residue modulo $m_a$ has a unique representative of degree less
than $|\Phi_a|$.
\end{proposition}

\begin{proof}
The nonzero vectors $E_\lambda e_a$, $\lambda\in\Phi_a$, are mutually
orthogonal, so equality on $e_a$ is equivalent to equality of all supported
values.  The factor theorem gives the divisibility assertion.  The vectors $E_\lambda e_a$
span $\mathcal K_a$, on which $p(H)$ and $q(H)$ have the same action.
Their exponentials therefore agree on $\mathcal K_a$.  Polynomial division
gives the unique representative.
\end{proof}

In particular, a real polynomial in $H$ can realize PST between $a,b$
if and only if they are strongly cospectral for $H$.
Necessity follows from \cref{lem:spectral-pst}; for sufficiency,
interpolate $r_-(\lambda)=(1-\sigma_\lambda)/2$ on $\Phi_a$.  Then
\[
 \exp(-\ii\pi r_-(H))e_a=e_b.
\]
See also \cite[Theorem~3.2]{Monterde2022}.  If the supported values are extended by zero to the unsupported
eigenvalues, interpolation on the full spectrum gives the projection
$\sum_{\lambda\in\Phi_a:\,\sigma_\lambda=-1}E_\lambda$.

Strong cospectrality gives
\begin{equation}\label{eq:half-weights}
 \sum_{\sigma_\lambda=1}\|E_\lambda e_a\|^2
 =\sum_{\sigma_\lambda=-1}\|E_\lambda e_a\|^2=\frac12.
\end{equation}
Indeed, the difference of these sums is $\langle e_a,e_b\rangle=0$
and their sum is $1$.

\begin{proof}[Proof of \cref{thm:optimal-transfer-time}]
Suppose $p(H)$ has PST at time $\tau$ and has local width at most $W$.
For $\lambda,\mu\in\Phi_a$ with $\sigma_\lambda=-\sigma_\mu$,
dividing the corresponding phase equations in
\eqref{eq:spectral-phases} gives
\[
 \tau\bigl(p(\lambda)-p(\mu)\bigr)\in\pi+2\pi\Z.
\]
Thus $\tau W\geq\pi$.  If $\tau=\pi/W$, the local width must be $W$.
Two values in the same sign class differ by an integer multiple of
$2\pi/\tau=2W$, and hence are equal.  The two resulting values have
opposite phases and are separated by exactly $W$.  They are therefore
$c$ on the positive class and $c+\varepsilon W$ on the negative class.
By \cref{prop:local-polynomial-residue}, this is equivalent to
\eqref{eq:all-optimal-polynomials}.

Conversely, interpolation gives these values, and every polynomial with
these values satisfies
\[
 \exp(-\ii\pi p(H)/W)e_a=\ee^{-\ii\pi c/W}e_b.
\]
The lower bound proves that $\pi/W$ is its first positive transfer time.
Equation~\eqref{eq:half-weights} gives weight $1/2$ at each of
the two supported eigenvalues of $p(H)$.
Since the supported values already have separation $W$,
$W(p(H))\leq W$ holds if and only if every unsupported value
lies between $c$ and $c+\varepsilon W$.
Such values may be assigned arbitrarily in that interval and interpolated
on the full spectrum.  This proves existence for the global constraint
as well as the classification.
\end{proof}

A polynomial attaining $\pi/W$ under the local width constraint has
local degree $\delta$.  Interpolation on the full spectrum gives a
polynomial of degree less than $|\spec(H)|$ attaining the same time
under the global width constraint.

\begin{proof}[Proof of \cref{thm:degree-width}]
Let $p(H)$ have local degree at most $d$, and replace $p$ by its
representative modulo $m_a$.  If $p(H)$ has PST at
$\tau$, dividing \eqref{eq:spectral-phases} at $\lambda_j$ by
the equation at $\lambda_0$ gives
\[
 z_j=\frac{\tau}{\pi}\bigl(p(\lambda_j)-p(\lambda_0)\bigr)
 \in\Z,\qquad
 z_j\equiv\frac{1-\sigma_j\sigma_0}{2}\pmod2.
\]
Thus $z\in\mathcal Z_d$.  Both signs occur, so every such integer vector
has positive range.  If $W_a(p(H))\leq W$, then
\[
 R_d\leq\max_jz_j-\min_jz_j
       =\frac{\tau}{\pi}W_a(p(H))\leq\frac{\tau W}{\pi}.
\]
This proves the lower bound in \eqref{eq:degree-optimum}.

If $\mathcal Z_d\ne\varnothing$, its positive integer ranges have a
least member.  Choose a minimizing vector $z$ and a polynomial $q$ of
degree at most $d$ with $q(\lambda_j)-q(\lambda_0)=z_j$.
The polynomial $p=(W/R_d)q$ has local width $W$.  At time
$\tau=\pi R_d/W$, its phase differences are $\pi z_j$, so
\eqref{eq:parity-lattice} gives \eqref{eq:spectral-phases} with
$\gamma=\sigma_0\exp(-\ii\tau p(\lambda_0))$.  Hence $p(H)$ has PST.
The lower bound just proved excludes an earlier transfer time and proves
attainment.

By \cref{thm:optimal-transfer-time}, attaining $\pi/W$ is equivalent to
the residue $c+\varepsilon Wr_-$.  Since both signs occur, $r_-$ is
nonconstant, and this representative has degree exactly $\delta$.
This proves $R_d=1$ if and only if $d\geq\delta$.  The remaining finite
values of $R_d$ are integers at least $2$.
\end{proof}

\Cref{thm:degree-width} is independent of the reference eigenvalue: replacing
$\lambda_0$ by $\lambda_\ell$ replaces $z_j$ by $z_j-z_\ell$.
The range is unchanged and
$z_j-z_\ell\equiv(1-\sigma_j\sigma_\ell)/2\pmod2$.  For any prescribed
range bound $R$, only finitely many integer vectors need be checked, since
$z_0=0$ forces $|z_j|\leq R$.  A search with a fixed range bound does not decide whether
$\mathcal Z_d$ is empty.

\begin{corollary}\label{cor:distance-degree}
Let $H=A(G)$ and let $a,b$ be strongly cospectral.  The least local degree
of a polynomial attaining transfer time $\pi/W$ under local width
at most $W$ satisfies
\[
 \delta\geq\operatorname{dist}_G(a,b).
\]
\end{corollary}

\begin{proof}
The polynomial $q=1-2r_-$ has degree $\delta$ and $q(A)e_a=e_b$.
For $j<\operatorname{dist}_G(a,b)$, the entry $(A^j)_{ba}$ is zero.
Thus a polynomial of smaller degree cannot have $(q(A))_{ba}=1$.
\end{proof}

\begin{example}\label{ex:P3}
For the endpoints of $P_3$, the supported eigenvalues are
$-\sqrt2,0,\sqrt2$, with signs $+,-,+$.  Here
$r_-(x)=1-x^2/2$, so $\delta=2$.  For an affine polynomial, the two
successive spectral gaps are equal; the parity constraints force their
normalized values to be odd integers.  Taking the smallest choice gives
an integer vector $(0,1,2)$, so $R_1=2$ and $R_2=1$.
Thus the factor-two gap in \cref{thm:degree-width} is sharp.
\end{example}

\begin{remark}\label{prop:paths}
The endpoints of $P_n$ are strongly cospectral and support all $n$
eigenvalues.  Indeed, the eigenvectors with coordinates
$\sin(\ell j\pi/(n+1))$ have nonzero endpoint coordinates related by
$(-1)^{j+1}$.  \Cref{cor:distance-degree} and interpolation give
$\delta=n-1$.
\end{remark}

\section{Degree constraints and spectral-weight partitions}
\label{sec:degree-graphs}

Throughout this section, $A$ is the adjacency matrix of a finite simple
connected graph, and $a,b$ are strongly cospectral vertices at distance
$D\geq2$ with $|\Phi_a|=D+1$.  Such support-size and distance conditions
are closely related to spectral extremality; see~\cite{Coutinho2016Extremal}.
Write
\[
 w_\lambda=(E_\lambda)_{aa},\qquad
 \Phi_\pm=\{\lambda\in\Phi_a:\sigma_\lambda=\pm1\},\qquad
 w(S)=\sum_{\lambda\in S}w_\lambda.
\]
Every $w_\lambda$ is positive, and $w(\Phi_+)=w(\Phi_-)=1/2$.
By \cref{cor:distance-degree} and interpolation, $\deg r_-=D$.
In particular, $R_{D-1}\geq2$ whenever it is finite, whereas $R_D=1$.

\begin{lemma}\label{lem:one-moment}
For $0\leq d<D$, real data $(y_\lambda)_{\lambda\in\Phi_a}$ are
values of a polynomial of degree at most $d$ if and only if
\begin{equation}\label{eq:moment-system}
 \sum_{\lambda\in\Phi_a}\sigma_\lambda w_\lambda
             \lambda^k y_\lambda=0\qquad(0\leq k<D-d).
\end{equation}
Consequently, $R_d$ is the minimum range of integer data satisfying
\eqref{eq:moment-system} and
$y_\lambda\equiv(1-\sigma_\lambda)/2\pmod2$.
For $d=D-1$, the system reduces to
\begin{equation*}
 \sum_{\lambda\in\Phi_a}\sigma_\lambda w_\lambda y_\lambda=0.
\end{equation*}
\end{lemma}

\begin{proof}
For $0\leq j<D$, no walk of length $j$ joins $a$ to $b$, so
\[
 \sum_{\lambda\in\Phi_a}\sigma_\lambda w_\lambda\lambda^j
 =(A^j)_{ba}=0.
\]
If $y_\lambda=q(\lambda)$ with $\deg q\leq d$, these identities
applied to $x^kq(x)$ prove necessity.  The $D-d$ row vectors in
\eqref{eq:moment-system} are linearly independent: their nonzero
coordinate factors $\sigma_\lambda w_\lambda$ multiply a Vandermonde
matrix.  Their common kernel has dimension $d+1$, equal to the
polynomial evaluation space, proving sufficiency.
Choose a reference node in $\Phi_+$.  Subtracting its even integer value
normalizes the reference value to zero without changing the range,
parities, or moment equations.  Now apply \cref{thm:degree-width}.
\end{proof}

\begin{proposition}\label{prop:degree-product}
Under the hypotheses of this section, $dR_d\geq D$ whenever
$1\leq d\leq D$ and $R_d$ is finite.  Thus
$T_d(W)\geq\lceil D/d\rceil\pi/W$.
\end{proposition}

\begin{proof}
Order the supported eigenvalues increasingly as $\lambda_0<\cdots<\lambda_D$.
The vector with entries
$\beta_j=1/\prod_{i\ne j}(\lambda_j-\lambda_i)$ annihilates evaluation
in degrees less than $D$: this follows by comparing the coefficient of
$x^D$ in the Lagrange interpolation formula.  That annihilator is
one-dimensional.  Hence $\sigma_jw_j$ is a nonzero scalar multiple
of $\beta_j$.  Since $w_j>0$ and the signs of $\beta_j$ alternate,
the signs $\sigma_j$ alternate too.

Let $q$ interpolate integer phase data of range $R$ and have degree
at most $d$.  Values at consecutive nodes have opposite parity.
For each $j$, choose a half-integer $h_j$ strictly between
$q(\lambda_j)$ and $q(\lambda_{j+1})$.  The intermediate value theorem
gives a root of $q(x)-h_j$ in $(\lambda_j,\lambda_{j+1})$.
There are $R$ possible half-integers, and each equation $q(x)=h$
has at most $d$ roots.  Since the $D$ intervals are disjoint,
$D\leq dR$.  Minimize over the phase data.
\end{proof}

For $1\leq d<D$, define the vector of moments
\[
 M_d(S)=\left(\sum_{\lambda\in S}w_\lambda\lambda^k
                         \right)_{k=0}^{D-d-1},\qquad
 H_d=M_d(\Phi_+)=M_d(\Phi_-).
\]
The equality follows from the walk identities in the proof of
\cref{lem:one-moment}.

\begin{theorem}\label{thm:weight-partition}
Under the hypotheses of this section, for $1\leq d<D$:
\begin{enumerate}[label=\textup{(\roman*)}]
\item $T_d(W)=2\pi/W$ if and only if there is a subset $S$ of one
sign class such that $M_d(S)=H_d/2$.
\item $T_d(W)\leq3\pi/W$ if and only if there are
$S_+\subseteq\Phi_+$ and $S_-\subseteq\Phi_-$ such that
\begin{equation}\label{eq:signed-quarter}
 M_d(S_+)-M_d(S_-)=H_d/2.
\end{equation}
\end{enumerate}
For $d=D-1$, the moment equalities in parts \textup{(i)} and
\textup{(ii)} reduce to $w(S)=1/4$ and
$w(S_+)-w(S_-)=1/4$, respectively.
The width constraint is local; both conclusions also hold for global
width when $\Phi_a=\spec(A)$.
\end{theorem}

\begin{proof}
For (i), integer data of range two have values in $\{m,m+1,m+2\}$.
If $m$ is even, the negative-sign data equal $m+1$, while the
positive-sign data are $m$ or $m+2$.  With $S$ the nodes of value
$m+2$, the spectral moment equations \eqref{eq:moment-system} become $2M_d(S)-H_d=0$.
If $m$ is odd, interchange the sign classes to obtain
$M_d(S)=H_d/2$ for a subset of $\Phi_-$.
Conversely, for such $S\subseteq\Phi_+$, prescribe $2$ on $S$,
$0$ on its complement in $\Phi_+$, and $1$ on $\Phi_-$.
For $S\subseteq\Phi_-$, prescribe $1$ on $S$, $-1$ on its complement
in $\Phi_-$, and $0$ on $\Phi_+$.  These data have range two and
satisfy \eqref{eq:moment-system} and
$y_\lambda\equiv(1-\sigma_\lambda)/2\pmod2$.  Their interpolant has degree
at most $d$ by \cref{lem:one-moment}.  Since $d<D=\deg r_-$,
range one is impossible, proving (i).

For (ii), subtract an even integer from admissible data of range at
most three so that their values lie in $[0,3]$ when their minimum is
even, and in $[-1,2]$ otherwise.  In the first case, let $S_+$ be the
positive-sign nodes of value $2$ and $S_-$ the negative-sign nodes
of value $3$.  The spectral moment equations \eqref{eq:moment-system} become
$2M_d(S_+)-H_d-2M_d(S_-)=0$.
In the second case, use the nodes of value $2$ in $\Phi_+$ and of
value $1$ in $\Phi_-$.  Their moment difference is $-H_d/2$;
complementing both subsets changes it to $H_d/2$.
Conversely, given \eqref{eq:signed-quarter}, prescribe values $2,0$
on $S_+,\Phi_+\setminus S_+$, and values $3,1$ on
$S_-,\Phi_-\setminus S_-$.  Apply \cref{lem:one-moment} and
\cref{thm:degree-width}.  The scalar specialization uses
$H_{D-1}=1/2$.
\end{proof}

\begin{corollary}\label{thm:graph-costs}
Let $a,b$ be strongly cospectral vertices of a bipartite graph, at
an even distance $D\geq2$, with $|\Phi_a|=D+1$.  Under local width $W$,
\begin{equation*}
 T_{D-1}(W)=2\pi/W,\qquad T_D(W)=\pi/W.
\end{equation*}
The same formulas hold under global width for the endpoints of
$P_{2r+1}$ and antipodal vertices of $Q_{2r}$, where $D=2r$.
\end{corollary}

\begin{proof}
Let $B$ be the diagonal matrix taking values $+1$ and $-1$ on the two
bipartition classes.  Then $BAB=-A$, so $BE_\lambda B=E_{-\lambda}$.
As $D$ is even, $a,b$ lie in the same class.  It follows that
\[
 w_{-\lambda}=w_\lambda,\qquad
 \sigma_{-\lambda}=\sigma_\lambda.
\]
The support is symmetric and has odd size $D+1$, so it contains zero.
Choose the sign class not containing zero.  Its positive eigenvalues
have half of that class's total weight, namely $1/4$.
\Cref{thm:weight-partition}(i) gives $R_{D-1}=2$, and $\deg r_-=D$
gives $R_D=1$.
For $P_{2r+1}$, the endpoint support is the full simple spectrum, as in
\cref{prop:paths}.  For $Q_{2r}$, the antipodal support is the full
spectrum of size $2r+1$, as verified at the start of
\cref{sec:equally-spaced}.  Both pairs have distance
$2r$, so the assertions about global width follow.
\end{proof}

\section{Equally spaced spectra and hypercubes}\label{sec:equally-spaced}

Label the vertices of $Q_n$ by $\{0,1\}^n$.  Its character vectors
$\chi_S(x)=(-1)^{\sum_{i\in S}x_i}$ are eigenvectors of eigenvalue
$n-2|S|$.  At antipodal vertices, their values differ by $(-1)^{|S|}$.
The support is the full spectrum, with
\begin{equation*}
 \lambda_j=n-2j,\qquad \sigma_j=(-1)^j,\qquad
 w_j=2^{-n}\binom nj\qquad(0\leq j\leq n).
\end{equation*}
In particular, the antipodal distance is $n$ and $\deg r_-=n$.
These are also the Hamming-scheme data used in
\cite[Appendix~A.3]{SongPowersParity}.

If $A_j$ is the distance-$j$ matrix of $Q_n$, then
\[
 AA_j=(n-j+1)A_{j-1}+(j+1)A_{j+1},
 \qquad A_{-1}=A_{n+1}=0.
\]
Induction shows that $A_j$ is a polynomial in $A=A_1$ of degree $j$,
with leading coefficient $1/j!$.  Hence the matrices $p(A)$ with $\deg p\leq d\leq n$ are exactly
the linear combinations of $A_0,\ldots,A_d$.
Thus a degree bound $d$ permits couplings only between vertices
at distance at most $d$, with equal couplings at each distance.

Write $x^{\underline\ell}=x(x-1)\cdots(x-\ell+1)$ and
$\Delta_s f(x)=f(x+s)-f(x)$, with $\Delta=\Delta_1$.
For a polynomial $f$ on the current finite set of nodes, write
$R(f)=\max f-\min f$.

\begin{lemma}\label{lem:integer-shift}
Let $f\in\R[x]$ have degree at most $d\geq1$ and take integer values
at $0,\ldots,d$.  For positive integers $m$ and nonnegative integers $t$,
\[
 \frac{m}{\gcd(m,L_d)}\ \bigm|\ f(t+m)-f(t).
\]
\end{lemma}

\begin{proof}
Newton interpolation gives
$f(x)=\sum_{k=0}^d c_k\binom{x}{k}$ with
$c_k=\Delta^k f(0)\in\Z$; in particular, $f$ is integer-valued
at all nonnegative integers.  For $1\leq i\leq d$, the identity
$i\binom mi=m\binom{m-1}{i-1}$ shows that
$m/\gcd(m,i)$ divides $\binom mi$.
Since $i\mid L_d$, the integer $g=m/\gcd(m,L_d)$ divides each
$\binom mi$, $1\leq i\leq d$ (including zero coefficients when $i>m$).
Vandermonde's identity gives
\[
 \binom{t+m}{k}-\binom tk
 =\sum_{i=1}^k\binom mi\binom{t}{k-i}.
\]
Every summand is divisible by $g$; summing with the integer
coefficients $c_k$ proves the claim.
\end{proof}

\begin{lemma}\label{lem:discrete-range}
Let $1\leq\ell\leq N$, and let $F$ be a real polynomial of degree
$\ell$ with leading coefficient $c/\ell!$.  Its range $R$ on
$0,\ldots,N$ satisfies
\begin{equation*}
 R^2\geq\frac{4c^2}{N+1}
 \frac{\binom{N+\ell+1}{2\ell+1}}{\binom{2\ell}{\ell}}.
\end{equation*}
\end{lemma}

\begin{proof}
Set
\[
 u(x)=x^{\underline\ell}(x-N-1)^{\underline\ell},
 \qquad G(x)=\Delta^\ell u(x).
\]
The leading coefficient of $G$ is $(2\ell)!/\ell!$.
Since $u$ vanishes at $0,\ldots,\ell-1$ and
$N+1,\ldots,N+\ell$, shifting the finite sums gives, for every polynomial $h$,
\[
 \sum_{x=0}^N G(x)h(x)
 =(-1)^\ell\sum_{t=\ell}^N u(t)\Delta^\ell h(t-\ell).
\]
Thus $G$ is orthogonal to every polynomial of degree less than $\ell$
for the counting measure on $0,\ldots,N$.  Taking $h=G$, and using
$\Delta^\ell G=(2\ell)!$, gives
\begin{align*}
 \sum_{x=0}^N G(x)^2
 &=(2\ell)!(\ell!)^2
   \sum_{t=\ell}^N\binom t\ell\binom{N+\ell-t}\ell\\
 &=(2\ell)!(\ell!)^2\binom{N+\ell+1}{2\ell+1}.
\end{align*}
The last equality is Vandermonde's convolution.
For every real constant $b$, the difference
$F-b-cG/(2\ell)!$ has degree less than $\ell$.  Orthogonality yields
\[
 \sum_{x=0}^N(F(x)-b)^2
 \geq\frac{c^2}{\binom{2\ell}{\ell}}
             \binom{N+\ell+1}{2\ell+1}.
\]
Choose $b$ to be the midpoint of the extreme values of $F$ on the nodes.
The left side is at most $(N+1)R^2/4$, proving the claim.
\end{proof}

\begin{proof}[Proof of the arithmetic bound in \cref{thm:arithmetic-width}]
For $z\in\mathcal Z_d$, the nodes $\lambda_j=\theta+hj$ give
$z_j=f(j)$ for a polynomial $f$ of degree at most $d$, with
$f(j)\in\Z$ and $f(j)\equiv j\pmod2$.
For any odd $m\leq n$, \cref{lem:integer-shift} shows that
$m/\gcd(m,L_d)$ divides $z_m-z_0$.  This difference is odd, hence nonzero
and its absolute value is at least $m/\gcd(m,L_d)$.
This proves \eqref{eq:arithmetic-width}.
\end{proof}

\begin{lemma}\label{lem:nonlinear-range}
Let $f\in\R[x]$ have degree $k\geq2$, with $k\leq n$, and suppose
$f(j)\in\Z$ and $f(j)\equiv j\pmod2$ for $0\leq j\leq n$.
Set $s=\lfloor n/k\rfloor$.  Then its range on these nodes satisfies
\begin{equation}\label{eq:nonlinear-range}
 \max_{0\leq j\leq n}f(j)-\min_{0\leq j\leq n}f(j)
 \geq \frac{s^k}{2^{k-2}}.
\end{equation}
If $n\geq3k$, this range is strictly greater than $n$.
\end{lemma}

\begin{proof}
In the Newton expansion $f(x)=\sum_{r=0}^k c_r\binom xr$,
the coefficient $c_k=\Delta^k f(0)$ of $\binom{x}{k}$ is a
nonzero even integer.
Indeed, modulo two the values of $f$ agree with the linear function $j$,
whose $k$th difference vanishes.  Thus $|c_k|\geq2$.
The $k$th difference with step $s$ is
\[
 \Delta_s^k f(0)=c_ks^k
 =\sum_{i=0}^k(-1)^{k-i}\binom ki f(is).
\]
All nodes $is$ lie between $0$ and $n$.  The positive and negative
coefficient sums on the right both equal $2^{k-1}$.  If $R$ is the
range of the data, this gives $2s^k\leq2^{k-1}R$, proving
\eqref{eq:nonlinear-range}.

If $s\geq3$, then
\[
 \frac{s^k}{2^{k-2}}\geq k(s+1)>n.
\]
For the first inequality, start with $s^2\geq2(s+1)$ at $k=2$,
and note that increasing $k$ multiplies the left side by $s/2\geq3/2$
and the right side by $(k+1)/k\leq3/2$.
The second inequality follows from $n<k(s+1)$.
\end{proof}

\begin{proof}[Proof of the rigidity statements in \cref{thm:arithmetic-width}]
The linear data $z_j=j$ give $R_d\leq n$.
Suppose admissible data of range at most $n$ are interpolated by a
polynomial $f$ of exact degree $k\leq d$.  The parities exclude $k=0$.
If $2\leq k\leq n/3$,
\cref{lem:nonlinear-range} gives $R(f)>n$, a contradiction.
For every $k\geq2$, the highest Newton coefficient is a nonzero
even integer.  Applying \cref{lem:discrete-range} therefore gives
\[
 R(f)\geq2\sqrt{\mathcal D_{n,k}}.
\]
Direct cancellation yields
\begin{equation*}
 \frac{\mathcal D_{n,k+1}}{\mathcal D_{n,k}}
 =\frac{(n+1)^2-(k+1)^2}{16(k+1)^2-4}.
\end{equation*}
For $k\geq\lfloor n/3\rfloor$, one has $n+1\leq3(k+1)$,
so $\mathcal D_{n,k+1}/\mathcal D_{n,k}<1$.  Consequently, if $k>n/3$ and
\eqref{eq:cube-finite-rigidity} holds, then
$R(f)\geq2\sqrt{\mathcal D_{n,k}}\geq
2\sqrt{\mathcal D_{n,d}}>n$, again a contradiction.
Thus $k=1$ whenever $n\geq3d$ or \eqref{eq:cube-finite-rigidity} holds.

To obtain \eqref{eq:cube-log-rigidity}, Stirling's formula gives,
uniformly for $d/n$ in a fixed compact subinterval of $(0,1)$,
\[
 \log\mathcal D_{n,d}=n\psi(d/n)-\log n+O(1).
\]
Differentiation gives $\psi'(u)=\log((1-u^2)/(16u^2))$.
The function is strictly decreasing on $[1/2,1)$, with
$\psi(1/2)=\tfrac12\log(27/16)>0$ and $\psi(1)=-2\log2$,
where the endpoint value is understood by continuity.  This also
justifies the definition of $\rho_0$.
Set
\[
 d_n=\left\lfloor\rho_0n-\left(\frac3{a_0}+\eta\right)\log n
                                                    \right\rfloor.
\]
Taylor expansion at $\rho_0$ shows that
\[
 \log\frac{4\mathcal D_{n,d_n}}{n^2}
 =n\psi(d_n/n)-3\log n+O(1)
 =a_0\eta\log n+O(1)+O((\log n)^2/n)\longrightarrow\infty.
\]
The monotonicity of $\mathcal D_{n,d}$ for $d>n/3$ and
\cref{lem:nonlinear-range} cover every smaller degree.  The assertion for a fixed $\beta<\rho_0$ follows as well.

Since $f$ is affine, $f(0)=0$ and the odd integer $f(1)$ give $f(j)=uj$
with $u$ odd.  Its range is $|u|n$, so $u=\pm1$ and $R_d=n$.
At time $\tau=n\pi/W$, normalized phase differences of any minimizing
representative $q$ therefore satisfy
\[
 (\tau/\pi)(q(\theta+hj)-q(\theta))=\pm j.
\]
Interpolation on $n+1$ nodes gives exactly
\eqref{eq:affine-minimizers}.  Conversely, these polynomials have
width $W$ on the support and transfer at $n\pi/W$.
\end{proof}

\begin{proof}[Proof of \cref{cor:exact-degree-growth}]
The affine spectral parametrization reduces the phase condition to
$f(j)\in\Z$ and $f(j)\equiv j\pmod2$ for $0\leq j\leq n$.
Let $A_{n,k}$ be the minimum range of such data with exact degree $k$.
The polynomial $f(j)=j+2\binom jk$ shows that this class is nonempty.
As in \cref{thm:degree-width}, well-ordering of the positive integer
ranges and width normalization give $S_{n,k}(W)=\pi A_{n,k}/W$.
In the Newton expansion
\[
 f(x)=\sum_{r=0}^k c_r\binom xr,
\]
$c_1$ is odd and $c_r$ is even for $r\ne1$.
In particular, the ordinary leading coefficient is $c_k/k!$, with
$|c_k|\geq2$.

The polynomial $\T_k(2t-1)$ has leading coefficient $2^{2k-1}$ and takes
alternating values at
\[
 t_i=\frac{1-\cos(i\pi/k)}2,\qquad 0\leq i\leq k;
\]
see~\cite[Sections~18.3 and~18.5]{DLMF}.
Put $\omega_i=1/\prod_{h\ne i}(t_i-t_h)$.
The signs of $\omega_i$ agree with $\T_k(2t_i-1)$, so the
leading-coefficient formula for Lagrange interpolation gives
\[
 \sum_i|\omega_i|=2^{2k-1}.
\]
Choose integers $j_i$ nearest to $nt_i$.  For fixed $k$ these are
distinct when $n$ is sufficiently large, and
\[
 \sum_i\left|\frac{1}{\prod_{h\ne i}(j_i-j_h)}\right|
 =n^{-k}\bigl(2^{2k-1}+O_k(n^{-1})\bigr).
\]
The interpolation weights sum to zero.  Subtracting the midpoint of
the extreme node values of $f$ therefore gives
\[
 \frac{|c_k|}{k!}
 \leq\frac{R(f)}2\,n^{-k}
             \bigl(2^{2k-1}+O_k(n^{-1})\bigr).
\]
Since $|c_k|\geq2$, this proves the required lower bound for $A_{n,k}$.

For the upper bound, let
\[
 h_n(x)=\frac{n^k}{2^{2k-2}k!}\T_k(2x/n-1).
\]
Its leading Newton coefficient is exactly two.  Round each lower
Newton coefficient to the nearest integer of the prescribed parity:
odd in degree one and even in all other degrees.  Each coefficient
changes by at most one.  The resulting polynomial $f_n$ is admissible,
has degree exactly $k$, and satisfies on the integer nodes
\[
 |f_n(j)-h_n(j)|\leq\sum_{r=0}^{k-1}\binom jr=O_k(n^{k-1}).
\]
Thus $R(f_n)\leq n^k/(2^{2k-3}k!)+O_k(n^{k-1})$.
Subtracting the even constant $f_n(0)$, if needed, normalizes the phase
data without changing their range.  This proves
\eqref{eq:exact-degree-growth}.

Suppose $R_n\leq(1+\varepsilon)n^k/(2^{2k-3}k!)$, where
$0\leq\varepsilon\leq1/4$.  The lower-bound argument gives
$|c_k|\leq2(1+\varepsilon)(1+O_k(n^{-1}))<4$ for all sufficiently large $n$.
Thus $|c_k|=2$.  Let $\epsilon_n=\operatorname{sgn}(c_k)$ and
put $G_n=\epsilon_nF_n$, with $F_n$ as defined before \cref{cor:exact-degree-growth}.
The values $G_n(j_i/n)$ have absolute value at most one.
Since $j_i/n\to t_i$ and the $t_i$ are distinct, the associated
Vandermonde matrices converge to an invertible matrix and have
uniformly bounded inverses.
Consequently, the coefficients of $G_n$ and
$\|G_n'\|_{\infty,[0,1]}$ are bounded by constants depending only on $k$.  Comparing any point to
a nearest grid node gives $|G_n(t)|\leq M_n=1+C_k/n$ on $[0,1]$.
The leading coefficient $a_n$ of $G_n$ satisfies
$a_n\geq2^{2k-1}/(1+\varepsilon)$.

Put $\kappa_k=2^{2k-1}=\sum_i|\omega_i|$.
Since $M_n-\operatorname{sgn}(\omega_i)G_n(t_i)\geq0$,
\[
 \sum_i|\omega_i|
       \bigl(M_n-\operatorname{sgn}(\omega_i)G_n(t_i)\bigr)
 =M_n\kappa_k-a_n\leq \kappa_k(C_k/n+\varepsilon).
\]
Each node value of $G_n$ therefore differs from the corresponding
value of $\T_k(2t-1)$ by $O_k(\varepsilon+n^{-1})$.
Interpolation at the fixed nodes $t_i$ proves
\eqref{eq:chebyshev-stability}.  If $R_n\sim n^k/(2^{2k-3}k!)$,
take $\varepsilon=\max\{0,R_n/(C_k^*n^k)-1\}$;
then $\varepsilon\to0$, proving \eqref{eq:chebyshev-profile}.

Finally, for $k=2$ subtract the even integer $f(0)$ and write
$f(j)=aj+bj(j-1)$, where $a$ is odd and $b$ is a nonzero integer.
Changing the sign if necessary, assume $b>0$.
If $b=1$, then $f(j)=(j-t)^2-t^2$ for an integer $t$.
Its minimum possible range on $0,\ldots,n$ is
$\lceil n/2\rceil^2$, attained by an integer $t$ nearest to $n/2$;
integers outside $[0,n]$ give larger ranges.
If $b\geq2$, put $r=\lfloor n/2\rfloor$.  The identity
\[
 (1-r/n)f(0)+(r/n)f(n)-f(r)=br(n-r)
\]
gives $R(f)\geq2\lfloor n^2/4\rfloor\geq\lceil n/2\rceil^2$
for $n\geq2$.  This proves \eqref{eq:quadratic-exact-degree}.
\end{proof}

The arithmetic bound also gives an exact statement for odd $n\geq3$:
if $p$ is its smallest prime divisor and $1\leq d<p$, then
$\gcd(n,L_d)=1$, so $R_d=n$, attained by the linear data $z_j=j$.
For composite odd $n$, this degree range is already contained in
$n\geq3d$; the arithmetic argument additionally handles the prime case
up to degree $n-1$.

If $p$ is an odd prime, $a\geq1$, and $1\leq d<p^a\leq n$, then
\begin{equation*}
 R_d\geq p^{\,a-\lfloor\log_p d\rfloor}.
\end{equation*}
For $a=1$ this recovers the prime obstruction $R_d\geq p$ for $d<p$.
For instance, the case $n\geq3d$ of \cref{thm:arithmetic-width} gives $R_d=25$ for $Q_{25}$
through degree $8$, and $R_d=35$ for $Q_{35}$ through degree $11$.

\begin{corollary}\label{thm:prime-cube}
For antipodal vertices of $Q_p$, where $p$ is an odd prime,
\begin{equation}\label{eq:prime-cube-cost}
 T_d(W)=\begin{cases}
 +\infty,&d=0,\\
 p\pi/W,&1\leq d<p,\\
 \pi/W,&d=p.
 \end{cases}
\end{equation}
These optima hold under both local and global width constraints.
\end{corollary}

\begin{proof}
For $n=p$, \cref{thm:arithmetic-width} gives $R_d\geq p$ for every
$1\leq d<p$.  The affine polynomial $(p-x)/2$ has values $z_j=j$
on the supported eigenvalues, so $R_d\leq p$ for all these degrees.
For $d=p$, the polynomial $r_-$ gives $R_p=1$; for $d=0$, a constant
Hamiltonian cannot transfer between distinct vertices.
\Cref{thm:degree-width} now proves \eqref{eq:prime-cube-cost}.
\end{proof}

For example, under width $W$, the affine Hamiltonian
$W(pI-A)/(2p)$ attains $p\pi/W$.  Increasing its allowed degree to any
value below $p$ cannot improve this time.  At degree $p$, the antipodal
permutation $A_p$ gives the Hamiltonian $W(I-A_p)/2$ and time $\pi/W$.
Thus the ratio $T_{p-1}(W)/T_p(W)=p$ is unbounded.
For $n\geq3$, the minimum transfer time for exact local degree two in
\cref{cor:exact-degree-growth} is at least $n\pi/W$.
Thus $R_1=R_2=n$: allowing quadratic terms does not improve the
optimum over affine Hamiltonians.

\section{Johnson graphs}\label{sec:johnson}

Let $A$ be the adjacency matrix of $J(2m,m)$, and let $a,b$ be complementary
$m$-subsets.  Their distance is $m$, and the support is the full
spectrum, with eigenvalues and signs
\begin{equation}\label{eq:johnson-data}
 \lambda_j=(m-j)^2-j,\qquad \sigma_j=(-1)^j,
 \qquad 0\leq j\leq m;
\end{equation}
see~\cite{VinetZhan2020} and \cite[Appendix~A.4]{SongPowersParity}.
Their weights are
\begin{equation}\label{eq:johnson-weights}
 w_j=\frac{\mu_j}{\binom{2m}{m}},\qquad
 \mu_j=\binom{2m}{j}-\binom{2m}{j-1},
\end{equation}
where $\binom{2m}{-1}=0$.
The one-dimensional nullspace of the transpose evaluation matrix for
degrees less than $m$ is generated by the vector with $j$th entry
$1/\prod_{k\ne j}(\lambda_j-\lambda_k)$, and direct multiplication gives
\[
 \prod_{k\ne j}(\lambda_j-\lambda_k)
 =(-1)^j\frac{j!(2m+1-j)!}{2m+1-2j}.
\]
The reciprocal of this product is $(-1)^j\mu_j/(2m)!$.  Thus the moment equations
in \cref{lem:one-moment}, together with
$\sum_j\mu_j=\binom{2m}{m}$, give \eqref{eq:johnson-weights}.

Put $r=m-j$, so that $\lambda_j=r(r+1)-m$.
Define polynomials in $x$ by
\[
 B_k(x)=\frac{1}{(2k)!}\prod_{i=0}^{k-1}(x-i(i+1)),
 \qquad B_0(x)=1.
\]
They satisfy $B_k(r(r+1))=\binom{r+k}{2k}$.
The evaluation matrix $(B_k(r(r+1)))_{0\leq r,k\leq m}$ is an
integer lower triangular matrix with unit diagonal.  Consequently, every polynomial of degree at most $m$
in $x$ that is integer-valued at $x=r(r+1)$, $0\leq r\leq m$,
has a unique expansion
$\sum c_kB_k$ with integer coefficients.

We use the following part of the Johnson-scheme classification
\cite[Theorems~3.6 and~3.7]{VinetZhan2020}: after subtracting its value
at $r=0$, integer phase data with parity $r$ are precisely
\begin{equation}\label{eq:johnson-phase-basis}
 f(r)=\sum_{k=1}^d c_k\binom{r+k}{2k},\quad c_k\in\Z,\qquad
 c_k\equiv
 \begin{cases}1,&k\text{ is a power of }2,\\
               0,&\text{otherwise}
 \end{cases}\pmod2,
\end{equation}
where $q\leq d\leq m$ and $q=2^{\lfloor\log_2m\rfloor}$.
There are no such phase data when $d<q$.
Reversing the index changes all signs $\sigma_j$ by the same factor,
so normalizing at $r=0$ gives $f(r)\equiv r\pmod2$.
In \eqref{eq:johnson-phase-basis}, $d$ bounds the degree in $x=r(r+1)$,
equivalently the local degree in $A$, rather than the degree in $r$.
In particular, under the degree bound $q$, the interpolating polynomial
in $x$ has degree exactly $q$, since $c_q$ is odd.

\begin{lemma}\label{lem:johnson-range-monotone}
For $\lceil m/2\rceil\leq k<m$,
$\Lambda_{m,k+1}\leq\Lambda_{m,k}/4$.
\end{lemma}

\begin{proof}
Put $N=2m+1$ and
\[
 D_\ell=\frac{4}{N+1}
        \frac{\binom{N+\ell+1}{2\ell+1}}{\binom{2\ell}{\ell}},
 \qquad 1\leq\ell\leq N.
\]
For $m\leq\ell<N$, direct cancellation gives
\[
 \frac{D_{\ell+1}}{D_\ell}
 =\frac{(N+\ell+2)(N-\ell)}{4(2\ell+3)(2\ell+1)}
 =\frac{(N+1)^2-(\ell+1)^2}{16(\ell+1)^2-4}
 \leq\frac14.
\]
Indeed, $N+1=2m+2\leq2(\ell+1)$, and
$3(\ell+1)^2\leq(16(\ell+1)^2-4)/4$.
Since $\Lambda_{m,k}^2=D_{2k}$, applying the ratio bound twice
proves the claim.
\end{proof}

\begin{proposition}\label{prop:johnson-rounding}
For $q=2^{\lfloor\log_2m\rfloor}$ there are admissible integer phase
data of degree $q$ and range at most $U_{m,q}$, where
\begin{equation}\label{eq:johnson-rounding-bound}
 U_{m,q}^2=\sum_{k=1}^q
       \frac{\binom{2m+2k+2}{4k+1}}{\binom{4k}{2k}}.
\end{equation}
They can be constructed by successive rounding of coefficients.
Moreover, uniformly for these values of $q$,
\begin{equation}\label{eq:johnson-upper-rate}
 \log U_{m,q}=m\log\Gamma+O(\log m),
 \qquad \Gamma=\frac{9+\sqrt{17}}8.
\end{equation}
\end{proposition}

\begin{proof}
Use the orthogonal polynomials $G_\ell$ from the proof of
\cref{lem:discrete-range}, with $N=2m+1$, and put
\[
 P_k(r)=\frac{G_{2k}(r+m+1)}{(4k)!},\qquad 1\leq k\leq q.
\]
Reflection in $N/2$ preserves the counting measure on $0,\ldots,N$.
Uniqueness of an orthogonal polynomial with prescribed leading
coefficient gives $G_{2k}(N-x)=G_{2k}(x)$.
Thus $P_k(-1-r)=P_k(r)$, and $P_k$ is a polynomial in $r(r+1)$.
Its leading coefficient as a polynomial in $r$ is $1/(2k)!$;
its coefficient of $B_k$ is therefore one.  The polynomials
$P_1,\ldots,P_q$ are mutually orthogonal on the reflected grid.
The norm calculation in \cref{lem:discrete-range} gives
\[
 \sum_{r=-m-1}^mP_k(r)^2
 =\frac{\binom{2m+2k+2}{4k+1}}{\binom{4k}{2k}}.
\]

Start with the zero polynomial and proceed through $k=q,q-1,\ldots,1$.
At step $k$, round the current coefficient of $B_k$ to the nearest
integer having the parity prescribed in \eqref{eq:johnson-phase-basis}.
Adding $\alpha_kP_k$, with $|\alpha_k|\leq1$, makes this change
without altering any coefficient already fixed.
The final polynomial is $F=\sum_{k=1}^q\alpha_kP_k$.
Subtract $F(0)$ to obtain admissible phase data $f$ with $f(0)=0$.
The degree is exactly $q$, since its coefficient of $B_q$ is odd.
Orthogonality gives
\[
 \sum_{r=-m-1}^m F(r)^2\leq U_{m,q}^2.
\]
Each value on $0,\ldots,m$ appears twice on the larger grid.
If $M$ and $L$ are the maximum and minimum of $F$ on $0,\ldots,m$, then
$(M-L)^2\leq2(M^2+L^2)\leq\sum_{r=-m-1}^mF(r)^2$.
Write $f(r)=P(r(r+1))$.  Since $\lambda_j=(m-j)(m-j+1)-m$,
the Hamiltonian $WP(A+mI)/R(f)$ has width $W$ and transfers
at time $\pi R(f)/W$.

For the asymptotic estimate, factorial bounds in
\eqref{eq:johnson-range-bound} give, uniformly for $1\leq k\leq m$,
\[
 \log\Lambda_{m,k}=m\psi(k/m)+O(\log m),
\]
with the continuous endpoint values $\psi(0)=0$ and $\psi(1)=-2\log2$.
Here one may use $\log(t!)=t\log t-t+O(\log(m+1))$ for $0\leq t\leq6m+2$,
with $0\log0=0$; fixed shifts of factorial arguments contribute
only $O(\log m)$.
The function $\psi$ has its unique maximum at $u_*=1/\sqrt{17}$, and
\[
 \psi(u_*)=\log\frac{1+u_*}{1-u_*}=\log\Gamma.
\]
Since $q>m/2$, the sum in \eqref{eq:johnson-rounding-bound} includes
an integer $k$ nearest to $mu_*$ for all sufficiently large $m$.
Bounding the sum above by $m$ times its largest term and below by
this one term proves \eqref{eq:johnson-upper-rate}.
\end{proof}

\begin{proposition}\label{prop:johnson-deficit}
Let $q=2^{\lfloor\log_2m\rfloor}$ and $1\leq t\leq m-q$.
There are integer phase data for complementary vertices of $J(2m,m)$
of degree at most $m-t$ and range at most
\begin{equation}\label{eq:johnson-deficit-upper}
 V_{m,t}:=6L_t\binom{2m+1}{t}\binom{4m+2}{t},
 \qquad L_t=\operatorname{lcm}(1,\ldots,t).
\end{equation}
They can be constructed from any integer phase data of degree at most
$m-t$ with $z_j\equiv j\pmod2$, by rounding coefficients in a real
basis of integer vectors satisfying \eqref{eq:moment-system}.
In particular,
\begin{equation}\label{eq:johnson-deficit-rate}
 \log R_{m-t}=O\!\left(t\log\frac{em}{t}\right)
\end{equation}
uniformly for $1\leq t\leq m-q$.  For fixed $t$ this gives
$R_{m-t}=O_t(m^{2t})$.  If $t=o(m)$ and $m-t\geq q$, then
$T_{m-t}(W)=(\pi/W)\exp(o(m))$.
\end{proposition}

\begin{proof}
Put $M=2m+1$, and let $\mathcal M$ be the $t\times(m+1)$ moment matrix
\[
 \mathcal M_{s,j}=(-1)^j\mu_j\lambda_j^s,
 \qquad 0\leq s<t,\quad0\leq j\leq m,
\]
with the nodes and multiplicities in
\eqref{eq:johnson-data} and \eqref{eq:johnson-weights}.
The matrix $\mathcal M$ has rank $t$.  By \cref{lem:one-moment}, the required data are
exactly the integer vectors $z\in\ker\mathcal M$ with $z_j\equiv j\pmod2$.
We normalize $z_0$ only after constructing the data.

For $0\leq i\leq m-t$, put $L=M-2i$ and define a vector $w^{(i)}$
supported on $i,\ldots,i+t$ by
\begin{equation}\label{eq:johnson-local-kernel}
 w^{(i)}_{i+r}
 =\frac{L_t}{\binom tr}
   \binom{i+r}{r}\binom{M-i-r}{t-r}
   \binom Lr\binom{L-r-t-1}{t-r},
 \qquad0\leq r\leq t.
\end{equation}
Every entry on this support is a positive integer.
Here $v_p$ denotes the exponent of a prime $p$ in an integer.
The factorial formula expresses
$v_p\binom tr$ as a sum of terms equal to zero or one,
with at most $\lfloor\log_p t\rfloor$ nonzero terms.
Thus $\binom tr$ divides $L_t$.
All binomial coefficients in \eqref{eq:johnson-local-kernel}
have nonnegative admissible arguments, since $L\geq2t+1$.

To verify the moment equations, write
$Q_i=\prod_{u=0}^{2t}(L-u)$.
Using
\[
 \lambda_j-\lambda_k=(k-j)(M-j-k),\qquad
 \mu_j=\frac{(2m)!(M-2j)}{j!(M-j)!},
\]
direct cancellation gives, for $0\leq r\leq t$,
\[
 (-1)^{i+r}\mu_{i+r}w^{(i)}_{i+r}
 \prod_{\substack{0\leq v\leq t\\v\ne r}}
       (\lambda_{i+r}-\lambda_{i+v})
 =\frac{(-1)^i(2m)!L_tQ_i}{t!\,i!\,(M-i-t)!}.
\]
The right side is independent of $r$.
The leading-coefficient formula for Lagrange interpolation on these
$t+1$ nodes therefore gives $\mathcal M w^{(i)}=0$.
The first nonzero coordinate of $w^{(i)}$ has index $i$.
The $m-t+1$ vectors are consequently linearly independent and form a
basis of $\ker\mathcal M$ over $\R$; they need not generate the entire integer
kernel.

The condition $m-t\geq q$ guarantees integer phase data $z^0$
by \eqref{eq:johnson-phase-basis}, after reversing the index and
subtracting the value at $j=0$.
Expand $z^0=\sum_i\alpha_iw^{(i)}$ over the real basis and choose an
even integer $b_i$ nearest to each $\alpha_i$.
Then
\[
 z=z^0-\sum_i b_iw^{(i)}
   =\sum_i(\alpha_i-b_i)w^{(i)}
\]
is still an integer vector in $\ker\mathcal M$ with the prescribed parities,
and $|\alpha_i-b_i|\leq1$.
This construction uses only rational arithmetic and nearest-integer
rounding, since the basis and $z^0$ are integer.

Vandermonde's convolution bounds the two pairs of binomial factors by
\[
 \binom{i+r}{r}\binom{M-i-r}{t-r}\leq\binom Mt,
 \qquad
 \binom Lr\binom{L-r-t-1}{t-r}\leq\binom{2M}{t}.
\]
At any fixed coordinate $j$, at most one vector contributes for each
value $r=j-i$.  For $t\geq2$, the interior binomial coefficients
are at least $t$, so $\sum_{r=0}^t\binom tr^{-1}\leq2+(t-1)/t<3$;
the sum is $2$ for $t=1$.  Hence
\[
 |z_j|\leq\sum_i w^{(i)}_j
 \leq3L_t\binom Mt\binom{2M}{t}.
\]
The range bound \eqref{eq:johnson-deficit-upper} follows.
Subtracting the even integer $z_0$ normalizes the phase data without
altering their range or degree.

To bound $L_t$, observe that the quotient $L_{2n}/L_n$ divides $\binom{2n}{n}$: each prime
power newly appearing between $n$ and $2n$ contributes a factor to
that binomial coefficient.  Iterating gives
$L_{2^a}\leq4^{2^a-1}$, and comparison with the next power of two gives
$L_t\leq16^t$.
Using $\binom Nt\leq(eN/t)^t$ in
\eqref{eq:johnson-deficit-upper} proves
\eqref{eq:johnson-deficit-rate}.
Finally, $(t/m)\log(em/t)\to0$ when $t=o(m)$.
The transfer-time bounds follow from \cref{thm:degree-width}.
\end{proof}

\begin{proposition}\label{prop:johnson-linear}
If $m\geq3$ is not a power of two, then complementary vertices of
$J(2m,m)$ admit integer phase data of degree at most $m-1$ and range
at most $3m-2$.  Consequently,
\begin{equation}\label{eq:johnson-linear-cost}
 T_{m-1}(W)\leq\frac{(3m-2)\pi}{W}.
\end{equation}
If $m$ is a power of two, no polynomial of local degree at most $m-1$
admits PST between complementary vertices.
\end{proposition}

\begin{proof}
Write $N=\binom{2m}{m}$ and use the integer multiplicities $\mu_j$
in \eqref{eq:johnson-weights}.  By \cref{lem:one-moment},
integer phase data have degree at most $m-1$ exactly when
\[
 \sum_{j=0}^m(-1)^j\mu_j z_j=0,
 \qquad z_j\in\Z,\quad z_j\equiv j\pmod2.
\]
Subtracting the even integer $z_0$ afterwards preserves the parities
and normalizes the reference value to zero.
The two sign classes each have total multiplicity $N/2$.
Thus, writing $z_j=(j\bmod2)+2x_j$, the equation is equivalent to
\[
 \sum_{j=0}^m(-1)^j\mu_jx_j=N/4.
\]
The identity $v_2(t!)=t-s_2(t)$ gives $v_2(N)=s_2(m)$, where $s_2(m)$
is the number of ones in the binary expansion of $m$.
Thus $N/4$ is an integer exactly when $m$ is not a power of two.
If $m$ is a power of two, the moment equation has no integer solution.

Suppose now that $4\mid N$, and let $h$ be an index maximizing
$\mu_j$.  The ratios
\[
 \frac{\mu_{j+1}}{\mu_j}
 =\frac{2m-2j-1}{2m-2j+1}\frac{2m-j+1}{j+1}
 \qquad(0\leq j<m)
\]
are decreasing in $j$.  In particular, $h\geq1$, $\mu_0=1$,
$\mu_1=2m-1$, and, for $j\geq1$,
\[
 \frac{\mu_{j+1}}{\mu_j}\leq\frac{\mu_2}{\mu_1}
 =m-1-\frac{1}{2m-1}<m-1.
\]
Construct integers $a_h,\ldots,a_0$ by successive nearest-integer
rounding.  Start with residual $N/4$; at step $j$, choose $a_j$
nearest to the current residual divided by $\mu_j$, then subtract
$a_j\mu_j$.  Set $a_j=0$ for $j>h$.
After each step the residual has absolute value at most $\mu_j/2$.
Since $\mu_0=1$ and the residual is integer, the last step makes it
zero.  Hence
\[
 \sum_{j=0}^m\mu_ja_j=N/4.
\]
For the first coefficient,
\[
 |a_h|\leq\frac{N}{4\mu_h}+\frac12\leq\frac{m+3}{4},
\]
because $N=\sum_j\mu_j\leq(m+1)\mu_h$.
For $1\leq j<h$, the ratio bound gives $|a_j|<m/2$.
For $m\geq5$, $(m+3)/4\leq(m-1)/2$; for $m=3$,
the integer $a_h$ has absolute value at most $1$.  Therefore
\[
 |a_j|\leq b:=\left\lfloor\frac{m-1}{2}\right\rfloor
 \quad(j\geq1),\qquad |a_0|\leq m-1.
\]
Now set $x_j=(-1)^ja_j$ and $z_j=(j\bmod2)+2x_j$.
The integers $z_j$ satisfy $z_j\equiv j\pmod2$ and
$\sum_{j=0}^m(-1)^j\mu_jz_j=0$.
For $j\geq1$ their values lie in $[-2b,1+2b]$, while
$|z_0|\leq2m-2$.  Their range is at most
$2m-1+2b\leq3m-2$.  Interpolation and
\cref{thm:degree-width} prove \eqref{eq:johnson-linear-cost}.
\end{proof}

\begin{proof}[Proof of \cref{thm:johnson-degree-loss}]
The uniform bound and its logarithmic form follow from
\cref{prop:johnson-deficit}; the bound for $t=1$ is
\cref{prop:johnson-linear}.
\end{proof}

\begin{remark}\label{rem:johnson-relaxation}
Let $M_{m,k}$ be the infimum of the range on $r=0,\ldots,m$ over
real polynomials $\sum_{j=0}^k c_jB_j(r(r+1))$ with $c_k=1$,
without integer or parity restrictions on $c_0,\ldots,c_{k-1}$.
\Cref{lem:discrete-range} and the polynomial $P_k$ from the proof of
\cref{prop:johnson-rounding} give
\[
 \Lambda_{m,k}\leq M_{m,k}
 \leq\sqrt{\frac{m+1}{2}}\,\Lambda_{m,k}.
\]
For the upper bound, use the norm formula for $P_k$ and the fact that
each node value occurs twice on the reflected grid.  Hence, whenever
$k/m$ remains in a compact subinterval of $(0,1)$,
\[
 \log M_{m,k}=m\psi(k/m)+O(\log m).
\]
In particular, $M_{m,k}$ is exponentially small if
$k/m$ tends to a number greater than $\rho_0$.  Improving the degree
threshold for PST beyond $\rho_0$ therefore requires information
beyond the highest coefficient alone.  The real polynomials in this
relaxation need not satisfy the integer phase conditions.
\end{remark}

\begin{proof}[Proof of \cref{thm:johnson-exponential}]
The feasibility assertion follows from
\eqref{eq:johnson-phase-basis}.  Under a degree bound $d\geq q$,
let $k\in[q,d]$ be the exact degree of an admissible phase polynomial
in $r(r+1)$.  Its degree as a polynomial $f$ in $r$ is $2k$, with
leading coefficient $c_k/(2k)!$.  Moreover, $f(-1-r)=f(r)$, so its
range on $-m-1,\ldots,m$ is the same as on $0,\ldots,m$.
Translate the larger grid to $0,\ldots,2m+1$ and apply
\cref{lem:discrete-range} to obtain $R\geq|c_k|\Lambda_{m,k}$.

If $d=q$, then $k=q$ and $|c_q|\geq1$.
If $d>q$ and $k>q$, there is no power of two between $q$ and $m$,
so $c_k$ is nonzero and even.  Hence
$R\geq2\Lambda_{m,k}\geq2\Lambda_{m,d}$ by
\cref{lem:johnson-range-monotone}.
Finally, if $d>q$ but $k=q$, \cref{lem:johnson-range-monotone} gives
$R\geq\Lambda_{m,q}\geq4\Lambda_{m,d}$.
These cases prove the lower bound in \eqref{eq:johnson-general-bounds}.

For the upper bound, choose $c_k=1$ when $k$ is a power of two and
$c_k=0$ otherwise.  These coefficients satisfy
\eqref{eq:johnson-phase-basis}.  The resulting data are nondecreasing
on $0,\ldots,m$, start at zero, and have range
\[
 B=\sum_{r=0}^{\log_2q}\binom{m+2^r}{2^{r+1}}.
\]
The bound $R_d\leq U_{m,q}$ and its exponential rate follow from
\cref{prop:johnson-rounding}; a degree-$q$ construction is permitted
for every $d\geq q$.

When $m=2q-1$ and $d=q$, simplifying the lower bound gives
\eqref{eq:johnson-exponential}.  The estimate for $U_{2q-1,q}$
is \eqref{eq:johnson-upper-rate}.
Stirling's formula gives
\[
 \frac{\binom{6q}{2q}}{\binom{4q}{2q}}
 \sim\frac{\sqrt3}{2}\left(\frac{27}{16}\right)^{2q}.
\]
Hence $R_q=\exp(\Theta(m))$ when $m=2q-1$.  Finally, $R_m=1$ by
\cref{thm:degree-width}, since the pair has distance $m$ and full
support of size $m+1$.
\end{proof}

\begin{proof}[Proof of \cref{cor:johnson-window}]
Uniformly for $u=d/m\in[1/2,\beta]$, Stirling's formula in
\eqref{eq:johnson-range-bound} gives
\begin{equation}\label{eq:johnson-uniform-asymptotic}
 \log\Lambda_{m,d}
 =m\psi(u)-\tfrac12\log m+O_\beta(1).
\end{equation}
Here the factorial arguments are bounded above and below by positive
constant multiples of $m$, so the error is uniform.
A direct differentiation gives
\[
 \psi'(u)=\log\frac{1-u^2}{16u^2}<0\qquad(1/2\leq u<1).
\]
Moreover, $\psi(1/2)=\tfrac12\log(27/16)>0$ and
$\lim_{u\uparrow1}\psi(u)=-2\log2<0$.
This proves the existence and uniqueness of $\rho_0$ and gives
$\psi(u)\geq\psi(\beta)>0$ in the stated interval.
Equations~\eqref{eq:johnson-general-bounds} and
\eqref{eq:johnson-uniform-asymptotic} now prove
\eqref{eq:johnson-window}.  The upper bound
$R_d\leq U_{m,q}=\Gamma^{m+o(m)}$ supplies the uniform exponential
upper estimate.

For the final assertion, suppose instead that
$\liminf d_m/m<\rho_0$.  Feasibility gives $d_m\geq q>m/2$.
There are therefore $\beta\in(1/2,\rho_0)$ and an infinite
subsequence with $d_m\leq\beta m$.  On this subsequence,
\[
 \log(\tau_mW_m/\pi)\geq\log R_{d_m}
 \geq m\psi(\beta)-\tfrac12\log m+O_\beta(1),
\]
contradicting the assumed subexponential growth.
\end{proof}

\begin{proposition}\label{prop:johnson-optima}
For complementary vertices, under either local or global width $W$,
\[
 \begin{array}{c|cc}
  \text{graph}&\text{degree bound }m-1&\text{degree bound }m\\ \hline
  J(6,3)&T_2(W)=2\pi/W&T_3(W)=\pi/W\\
  J(10,5)&T_4(W)=4\pi/W&T_5(W)=\pi/W
 \end{array}
\]
\end{proposition}

\begin{proof}
For $m=3$, the eigenvalues are $(9,3,-1,-3)$ and the multiplicities
$\mu_j$ are $(1,5,9,5)$.  Either negative-sign node has weight $5/20=1/4$,
so \cref{thm:weight-partition}(i) gives $R_2=2$.
Explicitly,
\[
 q(x)=\frac{-x^2+8x+9}{24}
\]
has values $(0,1,0,-1)$, and $(W/2)q(A)$ attains time $2\pi/W$.

For $m=5$, the eigenvalues and multiplicities are
\[
 (\lambda_j)=(25,15,7,1,-3,-5),\qquad
 (\mu_j)=(1,9,35,75,90,42).
\]
The total multiplicity is $252$.  The possible subset sums in the two
sign classes are, respectively,
\begin{align*}
 \mathcal A&=\{0,1,35,36,90,91,125,126\},\\
 \mathcal B&=\{0,9,42,51,75,84,117,126\}.
\end{align*}
No element of $\mathcal A$ differs from an element of
$\mathcal B$ by $63=252/4$.
By \cref{thm:weight-partition}(ii), $R_4>3$.
On the other hand, the data
\[
 (y_j)=(0,-1,0,-3,-4,-3)
\]
have the required parities, range four, and signed weighted sum zero,
since $9+225-360+126=0$.  By \cref{lem:one-moment} their interpolating
polynomial $q$ has degree at most four.  The Hamiltonian $(W/4)q(A)$
therefore attains $4\pi/W$, proving $R_4=4$.
Finally, in both cases $\deg r_-=m$, so $R_m=1$.
\end{proof}

\section{Further questions}\label{sec:conclusion}

For $m=2q-1$ with $q$ a power of two, $R_q$ grows exponentially
in $m$, while $R_{m-t}=\exp(o(m))$ for $t=o(m)$ and
$R_{m-1}=O(m)$.
Determining the smallest degree bound $d$ for which $R_d=\exp(o(m))$
when $m-d$ is proportional to $m$, the exact exponential growth
constants, and the full sequence $R_d$ remains open.
By \cref{rem:johnson-relaxation}, range estimates using only the
highest coefficient in the basis $B_k$ cannot improve the necessary
degree threshold $\rho_0$.  The bound
$R_{m-1}\leq3m-2$ does not decide whether $R_{m-1}=o(m)$.

For hypercubes, \cref{thm:prime-cube} determines the minimum transfer time for every degree bound in
odd prime dimensions.  In general dimensions, the affine rigidity range
extends to \eqref{eq:cube-log-rigidity}, but the first degree allowing
strict improvement over $n\pi/W$ is not determined.  

The optimization permits arbitrary real polynomial coefficients.
Requiring nonnegative off-diagonal entries or restricting individual
couplings leads to additional constraints on the same phase data.

\begingroup
\sloppy

\endgroup


\begin{thebibliography}{99}
\setlength{\itemsep}{0pt}

\bibitem{Bose2003}
S.~Bose,
Quantum communication through an unmodulated spin chain,
\emph{Phys. Rev. Lett.} \textbf{91} (2003), 207901.

\bibitem{ChristandlEtAl2004}
M.~Christandl, N.~Datta, A.~Ekert, and A.~J. Landahl,
Perfect state transfer in quantum spin networks,
\emph{Phys. Rev. Lett.} \textbf{92} (2004), 187902.

\bibitem{ChristandlEtAl2005}
M.~Christandl, N.~Datta, T.~C. Dorlas, A.~Ekert, A.~Kay, and A.~J. Landahl,
Perfect transfer of arbitrary states in quantum spin networks,
\emph{Phys. Rev. A} \textbf{71} (2005), 032312.

\bibitem{ChristandlVinetZhedanov2017}
M.~Christandl, L.~Vinet, and A.~Zhedanov,
Analytic next-to-nearest-neighbor $XX$ models with perfect state transfer and
fractional revival,
\emph{Phys. Rev. A} \textbf{96} (2017), 032335.

\bibitem{Coutinho2016Extremal}
G.~Coutinho,
Spectrally extremal vertices, strong cospectrality and state transfer,
\emph{Electron. J. Combin.} \textbf{23} (2016), Paper P1.46,
\href{https://doi.org/10.37236/5031}{doi:10.37236/5031}.

\bibitem{Godsil2011Periodic}
C.~Godsil,
Periodic graphs,
\emph{Electron. J. Combin.} \textbf{18} (2011), Paper P23, 15 pp.

\bibitem{Godsil2012}
C.~Godsil,
State transfer on graphs,
\emph{Discrete Math.} \textbf{312} (2012), 129--147.

\bibitem{Godsil2012When}
C.~Godsil,
When can perfect state transfer occur?,
\emph{Electron. J. Linear Algebra} \textbf{23} (2012), 877--890.

\bibitem{GodsilSmith2024}
C.~Godsil and J.~Smith,
Strongly cospectral vertices,
\emph{Australas. J. Combin.} \textbf{88} (2024), 1--21.

\bibitem{JafarizadehSufiani2008}
M.~A. Jafarizadeh and R.~Sufiani,
Perfect state transfer over distance-regular spin networks,
\emph{Phys. Rev. A} \textbf{77} (2008), 022315.

\bibitem{Kay2010}
A.~Kay,
Perfect, efficient, state transfer and its application as a constructive tool,
\emph{Int. J. Quantum Inf.} \textbf{8} (2010), 641--676.

\bibitem{Monterde2022}
H.~Monterde,
Strong cospectrality and twin vertices in weighted graphs,
\emph{Electron. J. Linear Algebra} \textbf{38} (2022), 494--518.

\bibitem{NessAlbertiSagi2022}
G.~Ness, A.~Alberti, and Y.~Sagi,
Quantum speed limit for states with a bounded energy spectrum,
\emph{Phys. Rev. Lett.} \textbf{129} (2022), 140403.

\bibitem{DLMF}
NIST Digital Library of Mathematical Functions,
Chapter~18: Orthogonal Polynomials, Sections~18.3 and~18.5,
\url{https://dlmf.nist.gov/18}, accessed 19 September 2026.

\bibitem{SongPowersParity}
X.~Song,
Perfect state transfer under matrix powers: parity and spectral arithmetic,
\href{https://arxiv.org/abs/2609.18018}{arXiv:2609.18018}, 2026.

\bibitem{VinetZhan2020}
L.~Vinet and H.~Zhan,
Perfect state transfer on weighted graphs of the Johnson scheme,
\emph{Lett. Math. Phys.} \textbf{110} (2020), 2491--2504.

\end{thebibliography}
\end{document}